\documentclass{bmathcryptology} 

\volume{1}
\issue{1}

\startpage{1}

\receiveddate{{\today}}
\revisiondate{{\today}}
\accepteddate{{\today}}

\articletitle{Trustless unknown-order groups}

\articleauthors{Samuel Dobson\aff{1}, Steven Galbraith\aff{1}, Benjamin Smith\aff{2}}

\articleaffiliations{
	\aff{1}University of Auckland, New Zealand\\
    \aff{2}Inria and Laboratoire d'Informatique (LIX), CNRS, École
    polytechnique, Institut Polytechnique de Paris, Palaiseau, France
}

\citationauthors{Dobson, S., Galbraith, S., \& Smith, B.} 


\MSClassification{94A60, 11Y40} 

\newif\ifllncs

\usepackage{bm}
\usepackage{amsmath}
\usepackage{amsfonts}
\usepackage{amssymb}
\usepackage{mathtools}
\usepackage{hyperref}
\usepackage{xcolor}
\usepackage{graphicx}
\usepackage{xifthen}
\usepackage{amsthm}
\usepackage{authblk}
\usepackage{multirow}

\usepackage[ruled,vlined,linesnumbered]{algorithm2e}
\DontPrintSemicolon
\SetKwProg{Fn}{function}{}{}
\SetKwFunction{Compress}{Compress}
\SetKwFunction{Decompress}{Decompress}
\SetKwFunction{Gen}{Gen}

\newcommand{\TODO}[1][]{\textbf{\color{red}\ifthenelse{\isempty{#1}}{TODO}{TODO: #1}}}

\newcommand{\Z}{\ensuremath{\mathbb{Z}}}
\newcommand{\F}{\ensuremath{\mathbb{F}}}
\newcommand{\Fbar}{\ensuremath{\overline{\mathbb{F}}}}
\newcommand{\Q}{\ensuremath{\mathbb{Q}}}

\newcommand{\softO}{\ensuremath{\widetilde{O}}}
\newcommand{\Jac}[1]{\ensuremath{J_{#1}}}
\newcommand{\Div}{\ensuremath{\operatorname{Div}}}
\newcommand{\bound}{\ensuremath{s(\lambda)}}

\DeclareMathOperator{\lcm}{lcm}

\newcommand{\hide}[1]{}

\newtheorem{assum}{Assumption}

\abstract{
    Groups whose order is computationally hard to compute have important applications
including time-lock puzzles, verifiable delay functions, and accumulators.
Many applications require trustless setup: that is,
not even the group's constructor knows its order.
We argue that the impact of Sutherland's generic
group-order algorithm has been overlooked in this context,
and that current parameters do not meet claimed security levels.
We propose updated parameters,
and a model for security levels 
capturing the subtlety of trustless setup.
The most popular trustless unknown-order group candidates
are ideal class groups of imaginary quadratic fields;
we show how to compress class-group elements
from $\approx 2\log_2(N)$ 
to $\approx \tfrac{3}{2}\log_2(N)$ bits,
where $N$ is the order.
Finally, we analyse Brent's proposal of Jacobians of hyperelliptic curves as unknown-order groups.
Counter-intuitively,
while polynomial-time order-computation algorithms
for hyperelliptic Jacobians exist in theory,
we conjecture that genus-$3$ Jacobians offer shorter keylengths
than class groups in practice.

}

\begin{document}

\begin{NoHyper}
\articleinformation 
\end{NoHyper}


\section{
    Introduction
}

Interest in groups of unknown order has been fuelled in recent years by
applications such as delay functions \cite{cryptoeprint:2018:712},
accumulators \cite{10.1007/978-3-030-26948-7_20}, and zero-knowledge
proofs of knowledge \cite{cryptoeprint:2019:1229}.  As the name
suggests, a group $G$ has \emph{unknown order} if it has a compact representation, but it is infeasible for
anyone to compute the order of $G$ efficiently without access to any secret
information used to construct $G$.
In the case of \emph{trustless setup}, the order should not even be known to the creator(s) of the group.
Some use-cases may require additional properties,
such as
the \emph{low order}
or the \emph{adaptive root assumptions}.
In order to be useful,
group operations in $G$ should be efficiently computable;
elements of $G$ should have a compact representation;
and it should be possible to efficiently sample random elements of~$G$.

Previously, there have been two proposals for concrete
unknown-order groups: RSA groups~\cite{rivest1996time}, and ideal class
groups of imaginary quadratic fields~\cite{2012/Lipmaa,Buchmann1988}.
Brent briefly suggested
hyperelliptic Jacobians as unknown-order groups~\cite{RichardP.:2000:PKC:891158};
but unlike RSA and class
groups, Jacobians have received little further attention.

RSA groups are
groups of the form $(\Z/N\Z)^{\times}$, where $N = pq$ is the product of two primes.
Computing the order of $(\Z/N\Z)^{\times}$ is equivalent to factoring~$N$.
A trusted party can efficiently generate an RSA modulus that resists all known attacks (including Sutherland's algorithm).
Sander~\cite{10.1007/978-3-540-47942-0_21} gave an algorithm to
trustlessly generate a modulus $N$ such that (with very high probability) $N$ has two large
factors---he calls this an
RSA-UFO (unknown factorisation object).
However, to match even the lower security of 1024-bit RSA moduli,
RSA-UFOs need ``bit length (much) greater than 40,000 bits'';
this is far too large to be efficient in most unknown-order group applications.

Class groups, on the
other hand, can be generated without trusted setup,
and so have received a lot of recent attention
(see e.g.~\cite{2012/Lipmaa,10.1007/978-3-030-17659-4_13,10.1007/978-3-030-26948-7_20}).
Buchmann and Hamdy~\cite{Buchmann01asurvey}
suggested that $1665$-bit discriminants (\(\approx 833\)-bit orders)
provide security equivalent to $3072$-bit RSA (i.e., 128-bit security);
more recently,
Biasse, Jacobson Jr., and Silvester~\cite{BJS10} claim that $1827$-bit discriminants
(\(\approx 914\)-bit orders)
are required to reach the same security level.

But the usual notions of security level are not appropriate
when evaluating class-group security for applications such as
accumulators, where the group is fixed.
The computational assumptions underlying security are not defined for a
fixed group, and there is no random self-reduction 
to show that all instances have the same security.
We argue that much larger group sizes are needed for secure
unknown-order groups in applications where the group is fixed for many users and used for a long period of time.

Precisely, we propose a new security model for unknown-order groups,
depending on two parameters $(\lambda, \rho)$.
\begin{definition}
    \label{def:lambda-rho}
    Let $\texttt{Gen}$ be an algorithm outputting groups.
    We say that $\texttt{Gen}$ reaches security level $(\lambda,\rho)$
    if
    with probability $1 - 1/2^\rho$ over the outputs of $\texttt{Gen}$,
    any algorithm $\mathcal{A}$ given an output $G$ of $\texttt{Gen}$
    requires at least $2^\lambda$ bit operations to
    succeed in computing $\#G$ with probability close to 1.
\end{definition}



A similar concept of security is implicit in~\cite{AlbrechtMPS18},
which considers the security of cryptosystems depending on a prime
system parameter $p$ provided by a possibly malicious party,
when in practice the users only verify the primality of $p$---and
thus ensure the security level of the system---up to a certain probability
(if at all).  One recommendation of~\cite{AlbrechtMPS18}
is that users ``ensure that composite numbers are wrongly identified as
being prime with probability at most \(2^{-128}\)'',
corresponding to \(\rho = 128\).

In our context, the probabilistic nature of security is not due to
malicious parties, or unreliable verification,
but rather to a fundamental mathematical fact:
the distribution of random Abelian group orders.
Our security definition is motivated by
Sutherland's generic group order algorithm~\cite{sutherland2007order},
whose runtime depends on the (unknown) order itself,
rather than the supposed size of the order.

The relevance of Sutherland's algorithm to cryptographic unknown-order groups
seems to have been overlooked,
but in~\S\ref{sec:class-group-security}
we show that it attacks some 
class groups with parameters that are widely considered secure. 
For example:
Sutherland's algorithm can compute the order of a class group with $1827$-bit  discriminant (i.e., 914-bit group order) 
in $\approx 2^{59}$ operations with probability ${\approx 2^{-58}}$, 
or in $\approx 2^{114}$ operations with probability $\approx 2^{-20}$. 
The problem is that among class groups with prime discriminants of a given size,
there is a set of weak instances \emph{depending on the order}.
A randomly-generated group is only vulnerable with small probability;
but since the order is unknown, we cannot check for vulnerability
without simply attempting to run the algorithm for the given time
(in contrast, in the RSA setting with trusted setup, the group order is
known to its generator, who can thus easily choose a group that is
not vulnerable).
For these reasons, we consider that $1827$-bit discriminants (and even the $2048$-bit discriminants suggested in~\cite{10.1007/978-3-030-26948-7_20,cryptoeprint:2019:1229}) do not meet the requirements for $128$-bit security.

We propose new group sizes in response,
depending on $(\lambda, \rho)$.
A paranoid choice, $(\lambda,\rho) = (128,128)$,
requires group sizes of around 3400 bits.
A more realistic (but still cautious) choice, $(\lambda , \rho) = (128, 55)$,
requires 1900-bit group sizes (and so 3800-bit discriminants; more than double the previous suggestion).
Table~\ref{table:parameters} gives sizes of group orders for various combinations of $(\lambda, \rho)$. An alternative would be to use multiple smaller groups in parallel, however we believe this approach is less efficient than working in a single larger group.

Our second major result is a more compact representation of class group elements.
Inspired by a signature-compression method of Bleichenbacher~\cite{Ble04},
in~\S\ref{sec:compression}
we compress elements of class groups to $3/4$ the size of the usual
representation---a particularly welcome saving
in the light of our updated, much larger class group parameters.

Our third and more speculative contribution, in~\S\ref{sec:concrete}, is
an analysis of Brent's proposal of subgroups of Jacobians of hyperelliptic genus-3 curves
as a source of unknown-order groups with trustless setup.
We find that Jacobians offer a distinct advantage over class groups at
the same security level:
the element representation size is smaller (2/3 the size if our new class group compression algorithm is used; if not, 1/2 the size),
since point compression for curves is optimal.
Using Jacobians also allows us to take advantage of the wealth of
algorithms for group operations and exponentiation
that have been developed and implemented
for hyperelliptic discrete-log-based cryptography,
which may be more efficient than their class-group equivalents
(though the lack of recent competitive implementations makes it
difficult to compare Jacobians and class groups in terms of real-world speed).

We acknowledge that there are, in theory, polynomial-time algorithms to
compute the group order of hyperelliptic Jacobians~\cite{pila1990frobenius,10.1007/10722028_18}.
However, there is evidence that these algorithms are already impractical
for discrete-log-based cryptographic group orders of around 256-bits,
let alone the much larger group orders that we have in mind.
While curves of any genus \(\ge 2\) might be considered, we suggest that
genus-3 curves are the best choice: their point-counting algorithms
are already very complex,
and their DLP is harder relative to higher-genus curves.
Naturally, if Schoof-type algorithms for genus 3 could be made efficient
over large prime fields,
then these groups would become insecure---but at least we have provided motivation for such work.

Some unknown-order group protocols
make stronger assumptions:
for example,
that finding elements of a given order is hard,
or that
extracting roots of a given element is hard.
In~\S\ref{sec:known-order} we consider
the problem of constructing points of known order
in class groups and Jacobians,
and explain how we might
work with Jacobians
when the low-order or adaptive root assumptions are imposed.


\paragraph{Acknowledgements.}
We thank Dan Boneh, Benjamin Wesolowski, Steve Thakur, and Jonathan Lee
for their valuable comments, feedback and suggestions on an earlier
version of this work.
We also thank Edward Chen, Luca De Feo, and Jean Kieffer
for beneficial discussion,
and the anonymous reviewers for their feedback and comments.

\paragraph{Notation.}
Recall that $\softO(x) = O((\log{x})^c \cdot x)$ for some constant $c$,
and for subexponential algorithms,
$L_x(\alpha) = \exp{[(1 +
o(1))(\log{x})^\alpha (\log{\log{x}})^{1-\alpha}]}$ for \(0 \le \alpha \le 1\).

\section{
    Sutherland's algorithm: the security of generic groups
}
\label{sec:sutherland}

Sutherland's \emph{primorial-steps}
algorithm~\cite[Algorithm 4.2]{sutherland2007order}
computes the order of an element in a generic group;
it can be used to probabilistically determine the exponent of a group.
Remarkably, it runs in $O(\sqrt{N/\log\log N}) = o(\sqrt{N})$ time
(where $N$ is the group order)
in the worst case,
but in fact the expected runtime depends heavily on the multiplicative structure of $N$.
The algorithm runs particularly quickly when $N$ is smooth,
which we do not expect (or desire!) in unknown-order groups;
but it also poses a significant threat to a larger class of groups.

Sutherland's algorithm is based on Shanks' baby-step giant-step
(BSGS) algorithm, but one can also use Pollard rho. Suppose we wish to compute the order of $\alpha$.
Instead of computing consecutive powers of $\alpha$ in the
baby-steps, we compute a new element $\beta = \alpha^E$
such that the order of $\beta$ is coprime to all primes
$2,3,\ldots,p_n \leq L$ for a chosen bound $L$,
by taking $E$ to be the product of the $p_i$,
each raised to an appropriate exponent $\lfloor \log_{p_i}(M) \rfloor$
(where $M$ is an upper bound of the group order).
The baby-steps are then all powers of $\beta$ with exponents coprime to $P_n$,
and the giant-step exponents are multiples of $P_n$,
where $P_n = \prod_{i=1}^np_i$.
As in BSGS, a collision allows $| \beta |$ to be learnt, which then allows $|\alpha |$ to be computed very efficiently.

Sutherland shows that if the order $N$ of $\alpha$ is uniformly distributed
over $[1,M]$ (for sufficiently large $M$) and $L = M^{1/u}$, then
this is an $O(M^{1/u})$ time and space algorithm that successfully computes $N$ with probability $P \geq G(1/u,2/u)$ \cite[Proposition 4.7]{sutherland2007order}.
Here $G(r,s)$ is the semismooth probability function
\[
G(r,s) = \lim_{x \to \infty} \psi(x, x^s, x^r) / x
\]
for \(0 < r < s < 1\),
where $\psi(x, y, z)$ is the number of integers up to $x$
semismooth with respect to $y$ and $z$ (that is, all prime factors less
than $y$, with at most one greater than $z$).

\begin{table}[ht]
    \centering
    \caption{Asymptotic semismoothness probabilities from \cite{sutherland2007order} and \cite{bach1996asymptotic}}
    \begin{tabular}{rl|rl|rl}
        $u$ & $G(1/u,2/u)$
        &
        $u$ & $G(1/u,2/u)$
        &
        $u$ & $G(1/u,2/u)$
        \\
        \hline
        \hline
        2.1 & 0.9488
        &
        5.0 & 0.4473
        &
        12.0 & 4.255e-12 \\
        \hline
        2.9 & 0.5038
        &
        6.0 & 1.092e-03
        &
        16.0 & 6.534e-19 \\
        \hline
        3.0 & 0.4473
        &
        10.0 & 5.382e-09
        &
        20.0 & 2.416e-26 \\
        \hline
    \end{tabular}
\label{table:gvals}
\end{table}

For each choice of $\rho$ (corresponding to the probability that a weak
group is generated), one must determine $u$ such that $1/2^\rho \approx G( 1/u, 2/u )$.
It then follows that the group size should be at least $u \lambda$ bits so that a $1/u$-th root attack requires at least $2^\lambda$ operations.
Table~\ref{table:gvals} gives some numerically computed values for $G(1/u,2/u)$
from~\cite{bach1996asymptotic} and~\cite{sutherland2007order}.
Using the method of Banks and Shparlinski~\cite{banks2006integers} to approximate the density of semismooth numbers,
we calculate that for a success probability of less than $2^{-100}$,
we should take $u = 22.5$; for $2^{-128}$, we should take $u = 26.5$.

As we mentioned in the introduction, taking $(\lambda,\rho) = (128,55)$
leads to \(\approx1920\)-bit group orders,
because $\rho = 55$ corresponds to $u = 15$, and $15 \cdot 128 = 1920$.
A more conservative choice would be $(\lambda,\rho) = (128,128)$,
corresponding to $u \approx 26.5$ and implying 3392-bit groups.
Table~\ref{table:parameters} gives sizes of group orders for various combinations of $(\lambda, \rho)$.
Once again we stress that our setting is different to the usual world of security levels.
We are dealing with a fixed class of weak instances of the computational problem.

%
%

\begin{table}[ht]
\centering
    \caption{Group size (bits) for various attack success probabilities
    \(2^{-\rho}\) and running costs \(2^\lambda\)}
 \begin{tabular}{c c | c c c c c c}

     & & \multicolumn{6}{c}{\(\rho\)} \\
                              & & $40$ & $55$ & $64$ & $80$ & $100$ & $128$ \\
 \hline
     \multirow{4}{*}{\(\lambda\)} & $55$ & $660$ & $825$ & $880$ & $1045$ & $1265$ & $1430$ \\
                              & $80$ & $960$ & $1200$ & $1280$ & $1520$ & $1840$ & $2080$ \\
                              & $100$ & $1200$ & $1500$ & $1600$ & $1900$ & $2300$ & $2600$ \\
                              & $128$ & $1536$ & $1920$ & $2048$ &
                              $2432$ & $2944$ & $3392$ \\

\end{tabular}
\label{table:parameters}
\end{table}

We remark that Sutherland's algorithm is less of a threat to unknown
order groups with trusted setup. For example, if there is an authority
that can be trusted to generate an RSA modulus $N = pq$ where $p$ and
$q$ are safe primes, then the order of $\Z_N^\times$ cannot be computed using Sutherland's approach.

\section{
    Ideal class groups as unknown-order groups
}

In this section we reconsider ideal class groups
as a source of trustless unknown-order groups.
We briefly recall the relevant background on class groups
in~\S\ref{sec:class-group-background};
detailed references include~\cite{Cohen:2010:CCA:1965598}
and~\cite{cox1989primes}.
We then reconsider class-group security
in~\S\ref{sec:class-group-security},
and give a new compression algorithm for class group elements
in~\S\ref{sec:compression}.

\subsection{Background on ideal and form class groups}
\label{sec:class-group-background}

An \textbf{imaginary quadratic field} is an algebraic extension
\[
K = \Q(\sqrt{d}) = \{a + b\sqrt{d} \, | \, a,b \in \Q\}
\]
where $d < 0$ is a square-free integer.
The discriminant $\Delta$ of $K$ is $d$ if $d \equiv 1 \pmod{4}$,
or $4d$ otherwise
(so $\Delta \equiv 0,1 \pmod 4$).
The ring of integers $\mathcal{O}_K$ is $\Z[\omega]$,
where $\omega = \frac{1}{2}(1+\sqrt{d})$ when $d \equiv 1 \pmod{4}$ and $\omega = \sqrt{d}$ otherwise.

The \textbf{ideal class group} is the quotient
$Cl(\mathcal{O}_K) = J_K/P_K$,
where $J_K$ is the (Abelian) group of non-zero fractional ideals of
$\mathcal{O}_K$, and $P_K < J_K$ is the subgroup of non-zero principal
fractional ideals.
In practice, we represent $Cl(\mathcal{O}_K)$
using the isomorphic \textbf{form class group} $Cl(\Delta)$
of binary quadratic forms of discriminant \(\Delta\)
(the discriminant of \(K\)).
We let $(a,b,c)$ denote the binary quadratic form
\[
    (a,b,c) = ax^2 + bxy + cy^2 \in \Z [x,y]
    \quad
    \text{with}
    \quad
    b^2 - 4ac = \Delta
    \,.
\]
We can represent this form
using only two coefficients $(a,b)$,
because $c$ is uniquely determined by $c = (b^2 - \Delta)/4a$.
A form \((a,b)\) is
\textbf{positive definite} if $a > 0$.
As with ideal classes,
there is an equivalence relation on quadratic forms:
$f$ and $g$ are equivalent if 
$f(x,y) = g(\alpha x + \beta y, \gamma x + \delta y)$
for some \(\alpha\), \(\beta\), \(\gamma\), and \(\delta\) in \(\Z\)
with \(\alpha\delta - \beta\gamma = 1\)
(that is, if they are in the same orbit under~\(\operatorname{SL}_2(\Z)\)).
Equivalent forms always
have the same discriminant.

We represent each equivalence class in $Cl(\Delta)$ using the unique
\textbf{reduced form} in the class.
A form $(a,b,c)$ is reduced if $|b| \le a \le c$; and if $|b| = a$ or $a = c$, then $b \ge 0$. 
Lagrange, and later Gauss and then Zagier, gave algorithms to find
the equivalent reduced form for any binary quadratic form.
The identity in $Cl(\Delta)$ is the equivalence class of the
form $(1,0,-k)$ if $\Delta = 4k$, or $(1,1,k)$ if $\Delta = 4k+1$.
The group law in $Cl(\Delta)$, known as composition of forms, is
due to Gauss;
this does not usually output a reduced form,
so reduction is an additional step.
We shall not give these algorithms here,
but refer the reader to~\cite{Cohen:2010:CCA:1965598}.

The order of $Cl(\mathcal{O}_K)$ is the
\textbf{class number} of $K$,
denoted $h(\Delta)$.
It follows from the Brauer--Siegel theorem
(see~\cite{10.1007/3-540-44448-3_18})
that for sufficiently large
negative discriminants,
on average the class number satisfies
\begin{equation}
    \label{eq:class-number-estimate}
    \log h(\Delta) \sim \log\sqrt{| \Delta |}
    \quad
    \text{as}
    \quad
    \Delta \to -\infty
\end{equation}
We can therefore conservatively assume
$\approx \frac{1}{2}\log_2{|\Delta}|$-bit group sizes for cryptographic-sized
negative discriminants.

The use of class groups in cryptography was first suggested by Buchmann
and Williams~\cite{Buchmann1988}.
Hafner and McCurley~\cite{Hafner1989ARS}
gave a sub-exponential \(L_{|\Delta|}(1/2)\) algorithm
for computing the order of $Cl(\Delta)$.
Thus,
the order of a class group $Cl(\Delta)$
of negative prime discriminant $\Delta \equiv 1 \pmod{4}$
is believed to be difficult to compute, if $\Delta$ is sufficiently large.
Lipmaa~\cite{2012/Lipmaa}
proposed that $Cl(\Delta)$ can be used as a group of unknown order
without trusted setup, simply by selecting a suitably large $\Delta$
and choosing an element in $Cl(\Delta)$ to be treated
as a generator
(it is not possible to know if it generates the whole of
$Cl(\Delta)$, or just a subgroup; we discuss this further below).

\subsection{The security of ideal class groups}
\label{sec:class-group-security}

Until now,
cryptographic class group parameters have mainly been proposed
with an eye to resisting subexponential algorithms
for computing orders of quadratic
imaginary class groups.
In this section we re-assess the security these parameters
in the light of Sutherland's algorithm,
and propose new (much larger) parameter sizes targeting the 128-bit
security level.

Hafner and McCurley gave their $L_{|\Delta|}(1/2)$ algorithm to compute the order of quadratic
imaginary class groups in 1989~\cite{Hafner1989ARS};
Buchmann extended this to compute the group structure and discrete logarithms.
See Biasse et al.~\cite{BJS10} for a more up-to-date evaluation of the security of ideal class groups.
The important thing to note is that
these algorithms all have the same subexponential complexity
$L_{|\Delta|}(1/2)$,
depending essentially on the size of $|\Delta|$.
In contrast, Sutherland's algorithm has exponential worst-case runtime,
but performs much faster with a non-negligible probability
depending on the structure of the class group---a factor that
Hafner--McCurley cannot exploit.
When computing the order of a random class group,
therefore, the small probability that Sutherland's algorithm outperforms
Hafner--McCurley must be taken into account.

The cryptographic parameter sizes in~\cite{10.1007/3-540-44448-3_18}
and~\cite{Buchmann01asurvey} both suppose that Hafner--McCurley
is the best known algorithm.
Concretely,
it is suggested that a $1665$-bit negative fundamental discriminant,
which means an approximately $833$-bit group order
(cf.~Eq.~\eqref{eq:class-number-estimate}),
should provide $128$-bit security.
Biasse, Jacobson Jr., and Silvester~\cite{BJS10}
improve on previous attacks and suggest $1827$-bit discriminants
(which implies \(\approx 914\)-bit orders) are needed to achieve 128-bit security.
These estimates have been quoted in more recent works,
including~\cite{cryptoeprint:2019:1229} which estimates that $1600$-bit
discriminants provide $120$-bit security, and
\cite{10.1007/978-3-030-26948-7_20} which proposes a slightly more
conservative discriminant size of 2048 bits for $128$-bit security.

Suppose we try to compute the order of a random class group
with $1827$-bit fundamental negative discriminant using Sutherland's algorithm.
Sutherland's algorithm has some important practical speedups
when specialized from generic groups to class groups---for example,
class group element negation is practically free,
so time and memory can be reduced by a factor of \(\sqrt{2}\)
(see~\cite[Remark 3.1]{sutherland2007order})---but these improvements do
not significantly impact security levels.
The performance of Sutherland's algorithm on a given quadratic imaginary class group
depends entirely on the class number.

Hamdy and M{\"o}ller~\cite{10.1007/3-540-44448-3_18} show that
imaginary class numbers
are more frequently smooth (although not significantly so) than
uniformly random integers of the same size.
We may therefore
conservatively approximate the smoothness probability of random class
group orders as being that of random integers.
With the results of~\S\ref{sec:sutherland},
the probability that a random class group with $1827$-bit fundamental
negative discriminant has less than 128 bits of
security ($u=7.1$) is at least $2^{-14.3}$, and the chance it has less
than $64$-bit security is $2^{-50}$. 
If a system is using a fixed class group as an accumulator,
then we need to ask if these probabilities of weakness are acceptable.
We claim that such groups do not satisfy Definition~\ref{def:lambda-rho} for $(\lambda,\rho) = (128, 128)$,
and so the security is weaker at these discriminant sizes than was previously thought.


Bach and Peralta~\cite{bach1996asymptotic} give $G(1/u, 2/u)$ for $u=20$
as $2.415504 \times 10^{-26} \approx 2^{-85}$. Thus, even for $85$-bit
security, we require $3400$-bit discriminants.
Using Banks and Shparlinski~\cite{banks2006integers}'s method of approximating $G(1/u,2/u)$,
we estimate that for $100$-bit security with respect to Sutherland's algorithm,
a discriminant of around $4500$ bits would be required.
For 128-bit security, $u=26.5$ gives $G(1/u, 2/u) \approx 2^{-128}$,
which implies a group order $N\approx  2^{128\times26.5} = 2^{3392}$,
and hence we estimate that $\Delta$ should be approximately $6784$ bits.
We emphasize that $G(1/u,2/u)$ is only a lower bound for the success
probability of Sutherland's algorithm,
but this should still serve at least as a guideline.

\subsection{Compressing ideal class group elements}
\label{sec:compression}

Bleichenbacher~\cite{Ble04} proposed a beautiful algorithm to compress
Rabin signatures. His method can also be used to compress elements
in ideal class groups. As far as we know, this simple observation has
not yet been made in the literature.

\subsubsection{Bleichenbacher's Rabin signature compression algorithm}
\label{sec:Bleichenbacher}

A Rabin signature on a message $m$ under the public key $N$
(an RSA modulus)
is an integer $\sigma$ such that
\[
    \sigma^2 \equiv m \pmod{N}
    \,.
\]
Normally $\sigma$ is the same size as $N$,
but Bleichenbacher showed how to bring this down to $\sqrt{N}$.
The continued fraction algorithm (or the Euclidean algorithm)
can be used to compute integers $s$ and $t$ with $0 \leq s < \sqrt{N}$ and $|t| \le \sqrt{N}$
such that $\sigma t \equiv s \pmod{N}$:
see Algorithm~\ref{alg:PartialXGCD} and
Lemma~\ref{lemma:PartialXGCD-correctness} below.
The compressed signature is \(t\).
Given \(m\) and \(t\), let \(x = mt^2 \bmod{N}\);
then \(s^2 \equiv x \pmod{N}\),
but $0 \leq s < \sqrt{N}$, so we can recover $s$ from $m$ and $t$
by taking the integer square root of \(x\);
and then it is trivial to recover $\sigma \equiv {s/t} \pmod{N}$
(note that if \(t\) is not invertible modulo \(N\),
then we have found a factor of \(N\),
and the signature scheme is broken).
We may therefore replace $\sigma$ with $t$,
which has half the bitlength of \(\sigma\).

\begin{algorithm}[ht]
    \caption{\texttt{PartialXGCD}.}
    \label{alg:PartialXGCD}
    \KwIn{Integers \(a > b > 0\)}
    \KwOut{
        Integers \(s\) in \([0,\sqrt{a})\)
        and \(t\) in \([-\sqrt{a},\sqrt{a}]\)
        such that \(s \equiv bt \pmod{a}\)
    }
    \( (s, s', t, t', u, u') \gets (b, a, 1, 0, 0, 1) \)
    \;
    \While(\tcp*[f]{Invariants: \(0 \le s < s'\), \(|u| \leq |u'|\), \(|t'| \leq |t|\), \(s = au + bt\), \(s' = au' + bt'\)}){\(s \geq \sqrt{a}\)}{
        \( q \gets s' \operatorname{div} s \)
        \tcp*{Euclidean division without remainder}
        \(
            (s, s', t, t', u, u')
            \gets
            (s' - qs, s, t' - qt, t, u' - qu, u)
        \)
        \;
    }
    \Return{\((s,t)\)}
    \;
\end{algorithm}

\begin{lemma}
    \label{lemma:PartialXGCD-correctness}
    Given integers \(a > b > 0\),
    Algorithm~\ref{alg:PartialXGCD}
    returns \((s,t)\)
    such that \(s \equiv bt \pmod{a}\),
    \( 0 \leq s < \sqrt{a} \), and
    \(0 < |t| \le \sqrt{a}\).
\end{lemma}
\begin{proof}
    Algorithm~\ref{alg:PartialXGCD} is a truncated version of the
    extended Euclidean algorithm,
    stopping when \(s < \sqrt{a}\)
    rather than \(s = 0\).
    The invariants
    \(s' > s \ge 0\), \(|u'| \geq |u|\), \(|t'| \leq |t|\),
    \(s = au + bt\), and \(s' = au' + bt'\)
    are easily verified;
    in particular, \(s \equiv bt \pmod{a}\).
    From \(|t'| \leq |t|\) and the fact that $t$ is initialized
    to $1$, we immediately get that $t \neq 0$.
    Another invariant,
    \(|s't| \le a\),
    is proven in~\cite[Lemma 2.3.3]{Galbraith:2012:MPK:2230462}.
    Since \(s\) takes a sequence of strictly decreasing values,
    at some point \(0 \le s < \sqrt{a}\) and \(s' \geq \sqrt{a}\);
    this is where the loop terminates.
    It remains to show that at this point,
    we also have
    \(|t| \le \sqrt{a}\):
    but this follows from the invariant \(|s't| \le a\)
    and \(s' \geq \sqrt{a}\).
\end{proof}

\subsubsection{An improved class group element compression algorithm}

Suppose we have a reduced form \((a,b,c)\) in \(Cl(\Delta)\).
Since $b^2 - 4ac = \Delta$ is a known constant, it suffices to store $(a,b)$.
Since the form is reduced,
we have $|b| \le a  < \sqrt{| \Delta |}$,
so the pair $(a,b)$ can be encoded in approximately $\log_2( | \Delta |)$ bits:
this is the traditional ``compressed'' representation of a class group element.

But we can do better.
Since \(b^2 - 4ac = \Delta\),
we have
\[
   b^2 \equiv \Delta \pmod{a}
\]
---a relation reminiscent of the Rabin signature verification equation.
The situation is not exactly the same---\(a\) is not an RSA modulus,
and \(b\) is in \((-a,a]\) rather than \([0,a)\)---but it is not difficult to adapt
signature compression to class group element compression,
encoding the coefficient \(b\) in half the space.

First, we reduce to the case where \(b \ge 0\):
we store the sign of \(b\) as \(\epsilon = 1\) if \(b < 0\), and \(0\)
otherwise,
and replace \(b\) with \(|b|\).
We will treat the special cases \(a = b\) and \(b = 0\)
later;
in the meantime, we may suppose that \(0 < b < a\).
Using Algorithm~\ref{alg:PartialXGCD},
we compute integers $s$ and $t$
such that $bt \equiv s \pmod{a}$, $0 \leq s < \sqrt{a}$, and $|t| \leq \sqrt{a}$;
then
\[
  s^2 \equiv \Delta t^2 \pmod{a}
  \,.
\]
Given $a$ and $t$, we can compute $x = \Delta t^2 \bmod{a}$,
and then \(x = s^2\) \emph{as integers} because $0 \le s^2 < a$;
so $s$ can be recovered as the exact (positive) integer square root.
Now \(bt \equiv s \pmod{a}\),
and the Bleichenbacher approach suggests
compressing \(b\) to \(t\) and
recovering \(b\) as \(s/t \pmod{a}\);
but since \(a\) is not an RSA modulus,
we may---and often do---have $\gcd(a,t) \not= 1$,
and then \(t\) cannot be inverted modulo~\(a\).

To fix this, we compress \((a,b)\)
to \((a',g,t',b_0,\epsilon)\),
where $g = \gcd(a,t)$, $a' = a/g$, $t' = t/g$, $b_0 = |b| \mod{f}$.
Here, $f \geq g$ is the smallest integer such that $\lcm(f, a') \geq a$,
and \(\epsilon\) and \(t\) are defined as above.
The reason for this choice of $f$ is that since $b$ is an integer
satisfying $0 < b < a$, it is not necessary to compute $b$ by computing
$b \pmod{a}$. Instead, we can recover $b \pmod{N}$ for any integer
$N \ge a$. Here we use $N = \lcm( f, a')$, where we ensure $N \geq a$,
and this $f$ is deterministically computable by the decompression
algorithm. This avoids a failure to recover $b$ uniquely, in the case
that $a'$ and $g$ share common factors, if we had simply used $f = g$.
Roughly, because two random numbers are coprime with probability
$6/\pi^2$ \cite{CJ88}, then we expect $f$ to only be a small
additive factor higher than $g$, and therefore that $\log_2{f} \approx \log_2{g}$.
Concretely, we compressed more than 30 million random class group
elements of 3845-bit discriminants using an implementation in python,
and found that $f-g$ had an average value of $0.1234$, and a maximum value
of $16$, after all these attempts. These numbers did not appear to grow
with the size of $\Delta$.

To decompress,
we compute $a = a'g$, $t = t'g$,
\(x = t^2\Delta \bmod{a}\),
and \(s = \sqrt{x}\).
Let $b' \equiv s'/t' \pmod{a'}$,
where $s' = s/g$; note that $s$ is always divisible by $g$ since
both $a$ and $t$ are, via the invariant $s = au + bt$
in Algorithm~\ref{alg:PartialXGCD}.
Then $b' \equiv b \pmod{a'}$,
and we can compute $b$ uniquely
from $b \equiv b_0 \pmod{f}$ and $b \equiv b' \pmod{a'}$
using the Chinese Remainder Theorem, since $b < \lcm(f,a')$.
If \(\epsilon = 1\) then we correct the sign,
replacing \(b\) with \(-b\);
and after computing \(c\) (if required) as \((b^2 - \Delta)/(4a)\),
we are done.

For the special case \(a = b\),
we exceptionally let \(t = 0\);
this is not ambiguous, because \(t = 0\)
cannot occur in any other case.
The compressed form is then \((a',g,t',b_0,\epsilon) = (1,a,0,0,0)\).
When $b = 0$,
we compress to \((a,0,0,0,0)\).
Again, this is unambiguous: no other element of \(Cl(\Delta)\)
compresses to this value.

Algorithms~\ref{alg:compress} and~\ref{alg:decompress}
make the compression and decompression procedures completely explicit.
Note that
\begin{align*}
    \log_2 a' + \log_2 g
    & =
    \log_2 a \approx \log_2\sqrt{|\Delta|}
    \shortintertext{and}
    \log_2 t' + \log_2 b_0
    & \le
    \log_2 t' + \log_2 f
    \approx
    \log_2 t \approx \tfrac{1}{2}\log_2\sqrt{|\Delta|}
    \,.
\end{align*}
Algorithm~\ref{alg:compress} therefore compresses the form $(a,b,c)$ to a
$\tfrac{3}{4}\log_2|\Delta|$-bit representation,
three-quarters of the size of the traditional $(a,b)$.
When a party receives a compressed group element it is necessary for them to execute Algorithm~\ref{alg:decompress} before performing group operations on the element.

\begin{algorithm}[ht]
    \caption{Element compression for $Cl(\Delta)$}
    \label{alg:compress}
    \KwIn{A reduced form \((a,b,c)\) in \(Cl(\Delta)\)
    (\(c\) may be omitted)}
    \KwOut{A compressed form $(a',g,t',b_0,\epsilon)$}
    \lIf{\(b = 0\)}{%
        \Return{\((a,0,0,0,0)\)}
    }
    \lIf{$a = b$}{%
        \Return{\((1,a,0,0,0)\)}
    }
    \(
        \epsilon
        \gets
        \begin{cases}
            1 & \text{if } b < 0 \\
            0 & \text{otherwise}
        \end{cases}
    \)
    \;
    \(b \gets |b|\)
    \;
    $(s,t) \gets \texttt{PartialXGCD}(a, b)$
    \tcp*{Now $s \equiv bt \pmod{a}$, with $0 \le s < \sqrt{a}$}
    $g \gets \gcd(a,t)$
    \;
    $(a',t') \gets (a/g,t/g)$
    \;
    $f \gets |g|$
    \;
    \While{$\lcm(f,a') < a$}{
        $f \gets f + 1$
    }
    $b_0 \gets b \bmod f$
    \;
    %
    %
    \Return{$(a',g,t',b_0,\epsilon)$}
    \;
\end{algorithm}

\begin{algorithm}[ht]
    \caption{Element decompression for $Cl(\Delta)$}
    \label{alg:decompress}
    \KwIn{A compressed form $(a',g,t',b_0,\epsilon)$ and $\Delta$}
    \KwOut{A reduced form \((a,b,c)\) in \(Cl(\Delta)\) (\(c\) may be omitted)}
    \lIf{\((g,t',b_0,\epsilon) = (0,0,0,0)\)}{%
        \Return{\((a',0,-\Delta/(4a'))\)}
    }
    \lIf{$t' = 0$}{%
        \Return{$(g,g,(g^2-\Delta)/(4g))$}
    }
    $(a,t) \gets (g\cdot a', g\cdot t')$
    \;
    $x \gets t^2\Delta \mod a$
    \;
    %
    %
    $s \gets \sqrt{x}$
    \tcp*{Integer square root}
    $s' \gets s/g$
    \tcp*{Exact integer division}
    $b' \gets s'\cdot t^{-1} \pmod{a'}$
    \;
    $f \gets |g|$
    \;
    \While{$\lcm(f,a') < a$}{
        $f \gets f + 1$
    }
    $b \gets \texttt{CRT}((b',a'),(b_0,f))$
    \tcp*{\(b \equiv b' \pmod{a'}\) and \(b \equiv b_0 \pmod{f}\)}
    \If{$\epsilon = 1$}{%
        $b \gets -b$
        \;
    }
    \Return{$(a,b,(b^2-\Delta))$}
\end{algorithm}

\section{
    Hyperelliptic Jacobians as unknown-order groups
}
\label{sec:concrete}

We now revisit Brent's proposal of hyperelliptic Jacobians
as a concrete source of unknown-order groups,
focusing on genus \(g = 3\).
Hyperelliptic Jacobians can be seen as the ideal class groups
of quadratic function fields.
We will argue that even despite the existence of theoretical polynomial-time
point-counting algorithms, these Jacobians may still present
a more efficient alternative to class groups at the same security levels.

\subsection{Hyperelliptic curves}

We briefly recall the relevant background here;
details can be found in~\cite{menezes1996elementary}
and~\cite{Galbraith:2012:MPK:2230462}.
The reader familiar with hyperelliptic curves may skip this section.

Let \(q\) be an odd prime power.
A hyperelliptic curve $C$ of genus $g > 1$ over \(\F_q\)
is defined by an equation $y^2 = f(x)$,
where $f$ is a monic, squarefree polynomial of degree $2g+1$ over $k$.\footnote{%
    We do not need the more general case of nonsingular curves $y^2 + h(x)y = f(x)$,
    where $\deg f = 2g + 1$ or $2g + 2$, and $\deg h \le g$.
}
A (finite) point on \(C\) is a tuple \((x,y)\) in \(\Fbar_q^2\) satisfying the
defining equation of \(C\);
there is also a unique point at infinity, denoted \(\infty\).
%



A \textbf{divisor} on $C$ is a formal sum of points $D = \sum m_P P$
where $m_P = 0$ for all but finitely many~$P$.
The degree of a divisor is $\deg{D} = \sum m_P$.
The divisors form a group
$\Div(C)$,
and the divisors of degree zero form a proper subgroup $\Div^0(C)$ of $\Div(C)$.
We let
$\mathcal{P}(C)$ denote the set of principal divisors of $C$:
that is,
divisors of the form $(\gamma) = \sum_{P \in C} \operatorname{ord}_P(\gamma) P$
for some $\gamma$
in the function field $\Fbar_q(C) = \Fbar_q(x)[y]/(y^2 - f(x))$
(here $\operatorname{ord}_P(\gamma)$ is the order of vanishing of
$\gamma$ at $P$).
Principal divisors have degree 0,
so $\mathcal{P}(C) \subset \Div^0(C)$;
the \textbf{Jacobian} is the quotient group
\begin{equation}
    \Jac{C} \cong \Div^0(C)/\mathcal{P}(C)
\end{equation}
(also known as the degree-0 Picard group, denoted by $\text{Pic}^0(C)$).

We compute with elements of $\Jac{C}$
using the Mumford representation~\cite{mumford}.
Each divisor class contains a unique \emph{reduced}
divisor in the form $P_1 + \cdots + P_r - r\infty$
(with the $P_i$ not necessarily distinct)
with $r \le g$ and $P_i \not= \tilde{P_j}$ for all $i \not= j$.
The reduced divisors correspond to pairs of polynomials
$\langle u(x), v(x) \rangle$, where $u$ is monic, $\deg{v} < \deg{u} \leq g$
and $v^2 \equiv f \pmod{u}$.
The roots of $u(x)$ are the $x$-coordinates of the points in the support of the divisor.
The divisor classes defined over \(\F_q\)---that is,
such that \(u\) and \(v\) have coefficients in \(\F_q\)---form a finite group,
denoted \(J_C(\F_q)\).
The Hasse--Weil bound tells us that
$\#\Jac{C} \sim q^g$; more precisely,
$$
    (\sqrt{q} - 1)^{2g} \leq \# \Jac{C}(\F_q) \leq (\sqrt{q} + 1)^{2g}
    \,.
$$

The group law on $\Jac{C}(\F_q)$
can be computed using Cantor's algorithm~\cite{cantor1987computing}
(see also \cite{Costello:2011:GLC:2186839.2186848}).
Efficient explicit formulae exist for $g = 2$ (see~\cite{Lange2005})
and
$g = 3$ (see~\cite{fan2007efficient}).


\subsection{The security of hyperelliptic Jacobians}
\label{sec:hyperelliptic-security}

Let $C$ be a hyperelliptic curve of genus $g$ over $\F_q$,
where $q = {p^n}$,
and \(\Jac{C}\) its Jacobian.
Recall that \(\#\Jac{C}(\F_q) \approx q^3\).
For \(\Jac{C}(\F_q)\) to be useful as an unknown-order group,
calculating \(\#\Jac{C}(\F_q)\) should be infeasible.
Besides generic algorithms,
two classes of algorithms specific to hyperelliptic Jacobians
are relevant here:
point-counting
and discrete-log algorithms.

As a base-line,
$C$ must be chosen
such that $\Jac{C}(\F_q)$ resists Sutherland's algorithm
with acceptable probability, as in~\S\ref{sec:sutherland}.
Sutherland's algorithm
has some important practical optimizations when specialized to
hyperelliptic Jacobians---for example,
we can exploit the fact that negation is effectively free
to decrease storage and runtime by a factor of up to $\sqrt{2}$
(see~\cite[Remark 3.1]{sutherland2007order}),
and even if a Jacobian is not directly vulnerable to Sutherland's algorithm,
its order may be deduced from that of vulnerable twists,
as in~\cite{sutherland2009generic}---but
these improvements do not significantly impact security levels.

To reach acceptable security levels against Sutherland's algorithm using genus 3 curves
then \(q\) must be subexponentially large.
Looking at Table~\ref{table:parameters},
the cautious \((\lambda,\rho) = (128,55)\) level
requires 1920-bit groups, or \(q \approx 2^{640}\);
the paranoid \((128,128)\) level
requires 3392-bit groups, or \(q \approx 2^{1131}\).
Fields of this size also address the concerns of Lee~\cite{cryptoeprint:2020:289}.


\subsubsection{Point-counting algorithms}
\label{sec:point-counting}

Computing $\#\Jac{C}(\F_q)$
is a classic problem (called ``point counting'') in algorithmic
number theory: the goal is to compute the zeta function of $C$,
from which we immediately get $\#\Jac{C}(\F_q)$.
The many dedicated point-counting algorithms fall naturally into two broad classes:
\(p\)-adic algorithms
and \(\ell\)-adic ``\emph{Schoof-type}'' algorithms.
The \(p\)-adic algorithms
(notably Kedlaya's
algorithm~\cite{kedlaya2001counting}
and its descendants~\cite{10.1093/imrn/rnm095})
have complexity polynomial with respect to \(g\) and \(n\),
but exponential in \(\log p\).
Taking \(q = p\), we can ignore these algorithms.

\hide{
The $p$-adic point-counting algorithms,
most notably Kedlaya's algorithm~\cite{kedlaya2001counting} and its descendants,
generally compute the action of Frobenius on some $p$-adic cohomology
group connected with $C$.
Kedlaya's original algorithm computes the zeta function of $C$
in time \(\softO(pg^4n^3)\).
Harvey~\cite{10.1093/imrn/rnm095} has reduced the dependence on $p$ to $p^{1/2}$.
For fixed, or very small $p$ (polynomial in $n$ and $g$),
this is polynomial-time;
in practice, this is highly efficient for $p$ into the thousands.
However, as $p$ grows larger,
these algorithms become completely impractical.
If we focus on Jacobians of curves $C$ over $\F_{p^n}$
with $p$ large---and especially, on curves over $\F_p$---then
we can safely ignore the threat of any Kedlaya-style algorithm.
}

Schoof-type algorithms
compute $\#\Jac{C}(\F_q)$
in polynomial time \emph{for fixed $g$}.
Indeed, from a theoretical point of view,
the existence of Schoof-type algorithms
may make the use of hyperelliptic Jacobians
as unknown-order groups seem perverse.
But Schoof-type algorithms are totally impractical
over large prime fields, even in genus as small as 3.
To understand why, we need to look at how they operate.

First, consider the case of elliptic curves ($g = 1$).
Schoof's ground-breaking \(\softO(\log^5q)\)
algorithm~\cite{schoof1985elliptic},
the first polynomial-time point-counting algorithm for elliptic
curves,
computes the characteristic polynomial of Frobenius
for a series of small prime $\ell$,
using polynomial arithmetic modulo the division polynomials $\Psi_\ell$,
before combining the results with the Chinese Remainder Theorem
to compute the zeta function.\footnote{
    Recall that if $E$ is an elliptic curve,
    then for all positive integers $n$
    there exists a division polynomial~$\Psi_n(X)$
    such that $\Psi_n(x(P)) = 0$ if and only if $P \not= 0$ is in the
    $n$-torsion of $E$.
    For odd prime \(\ell\),
    the degree of \(\Psi_\ell\) is \((\ell-1)/2\).
}
Its successor, the Schoof--Elkies--Atkin (SEA) algorithm~\cite{schoof1995},
runs in time \(\softO(\log^4q)\),
and has made elliptic-curve point counting routine.

Pila generalised Schoof's algorithm to higher-dimensional Abelian
varieties~\cite{pila1990frobenius}, including all Jacobians of curves.
Pila's algorithm is polynomial-time in $p$ and $n$,
but badly exponential in~$g$;
as far as we know, it has never been implemented.
The task gets a little simpler
when we specialize from general Abelian varieties
to hyperelliptic Jacobians.
The crucial objects are
the analogue of the division polynomials:
these are multivariate \emph{division ideals} vanishing on coordinates of points in torsion subgroups
(in this sense, \(\Psi_\ell\)
generates the \(\ell\)-division ideal in genus \(g = 1\)).
Cantor constructs generators for the $\ell$-division ideal
in~\cite{Cantor1994} (see also~\cite{cohen2005handbook});
in genus $3$,
the degrees of Cantor's $\ell$-division polynomials are bounded by $O(\ell^2)$
(see~\cite{abelard2018counting}).

Schoof-type point-counting is already challenging in genus~2.
Several genus-2 algorithms have been implemented and analysed,
beginning with Gaudry and Harley~\cite{10.1007/10722028_18}
and Gaudry and Schost~\cite{10.1007/978-3-540-24676-3_15}.
Pitcher's PhD thesis~\cite{pitcher} gave a genus-2
algorithm with complexity $O((\log q)^7)$.
Gaudry and Schost~\cite{gaudry2012genus} used an improved algorithm,
with a mixture of Pitcher's approach and exponential birthday-paradox algorithms,
to find a curve of secure order over the 127-bit Mersenne prime field~$\F_{2^{127}-1}$.
In their experiments, they claimed around 1,000 CPU hours on average
to compute the order of a random genus-2 curve over this 127-bit field.
Computing $\ell$-division ideals and analysing the action of
Frobenius on them
can become impractical for even moderately small $\ell$:
the computations mentioned above, with an 8 GB limit on RAM, used primes~\(\ell \le 31\)
(the earlier~\cite{10.1007/978-3-540-24676-3_15} used~\(\ell\) up to~\(19\)).
They also used small prime powers \(\ell^k = 2^{16}\), \(3^6\), \(5^4\), and \(7^2\).
These \(\ell\) and \(\ell^k\) are not sufficient to determine the group order;
to finish the order computation, they used one- or two-dimensional
random walks (a low-memory square-root algorithm:
see~\cite{2004/Gaudry--Schost:MCT} for details).
We underline the fact that finishing this point-counting computation is a situation where
in practice, an exponential algorithm is more practical than a polynomial-time one!

\hide{
Steven: I agree with the below.
\TODO[Ben: I've changed this to ``not sufficient to determine the group
order'', because it's kind of complicated.  What they're actually trying
to compute is \((s_1,s_2)\), where \(s_1\) is a 64-bit number and
\(s_2\) is a 127-bit number.  The primes \(11\) through \(31\) give them
about 31 bits of information on \((s_1,s_2)\), and then the small prime
powers can give them another 49 bits.  Best case, they've determined
\(s_1\) and they're left with a chunk of \(s_2\) to find using an
exponential random-walk algorithm.
Worst case, there are still a few bits of \(s_1\) left and they need a
two-dimensional variant of the random-walk algorithm.  But I think it's
tricky to say that they've determined, say, 80 bits of the group order
because it's more like they've determined 80 bits of \(s_2\): there are
only about 50 bits left to go until they've determined the whole group
order, rather than 175 bits or something, which is what it would be if
the group order were a random 256-bit number and not something drawn
from the Hasse--Weil interval...  TL;DR: I think it's complicated, and
I'm not sure we want to go into that much detail right here.]
}

We have found no practical work for general genus-2 curves going beyond \(\ell = 31\) in the literature.
Abelard's PhD thesis~\cite{abelard2018counting} discusses the
feasibility of continuing with larger~\(\ell\).
With time complexity $\softO(\ell^6 \log q)$
and space complexity $\softO(\ell^4 \log q)$,
running time becomes
more of an issue than memory: for 192-bit \(q\),
the computation for $\ell = 53$ could take around 1000 CPU-days,
yet still leave a search space of $\approx 2^{95}$ elements
in the exponential ``collision'' step of the algorithm.

This practical work has not been extended to genus~3.
The main obstruction is the complexity of computing with division ideals.
Some theoretical analysis and projected
complexities appear in~\cite{AGS-FOCM}.
Some first steps have been made for the
very special class of genus-3 Jacobians with known and efficiently
computable real multiplication endomorphisms in~\cite{AGS-ANTS},
following the analogous genus-2 algorithm
in~\cite{Gaudry--Kohel--Smith},
but this approach does not apply to general genus-3 hyperelliptic Jacobians.

Concretely,
taking $q \sim 2^{100}$ in genus $3$ would appear sufficient to resist
point counting on most curves \(C\),
and result in \(\#\Jac{C}(\F_q) \sim 2^{300}\);
the much larger group and field sizes required to resist Sutherland's
algorithm
render point counting irrelevant as an attack.

The upshot is that while point-counting for \emph{fixed} genus $g > 2$
is polynomial-time in theory,
it remains impractical---even infeasible---in the real world.
This is already true of the relatively small field sizes
relevant to discrete-logarithm-based cryptography;
it is even more so for the much larger, subexponential-sized fields
required to protect against Sutherland's algorithm in the unknown-order setting.

\begin{remark}
    Some work has been done on generating genus-2 and genus-3 Jacobians
    \emph{with a known number of points} using CM theory,
    for applications in DLP-based cryptography,
    notably by Weng~\cite{weng3,weng2003constructing}
    (see also e.g.~\cite{gaudry2012genus} and
    \cite{hess2001two}).
    Obviously, these curve-generation methods must be avoided for unknown-order
    applications.
\end{remark}

\begin{remark}
    One might hope that progress in computing higher-genus
    modular polynomials might yield a SEA analogue
    improving substantially on pure Schoof-style point
    counting.
    However,
    any SEA analogue in genus \(g > 3\)
    would actually be \emph{slower} than pure Schoof.
    Indeed,
    the number of isogenies splitting \([\ell]\)
    (and hence the degree of the ideal
    that a SEA analogue would use at the prime \(\ell\))
    is in \(O(\ell^{g(g+1)/2})\);
    this exceeds the degree of the \(\ell\)-division ideal,
    which is in \(O(\ell^{2g})\).
    Even for \(g = 3\),
    the asymptotic complexity of SEA is no better than that of Schoof.
\end{remark}

\subsubsection{Discrete logarithm algorithms}
\label{sec:discrete-logs}

If the DLP can be efficiently solved
in a subgroup
\(\langle{G}\rangle\subset\Jac{C}(\F_q)\),
the order of \(\langle{G}\rangle\)
can also be efficiently computed: if $xG =
\mathcal{O}$, where
is the cryptographic subgroup,
then $x$ is (a multiple of) the order.
Suppose, then, that we want to solve the DLP in $\Jac{C}(\F_q)$,
where $C$ is a curve of genus $g$ over $\F_q$.
Gaudry et al.~\cite{10.2307/40234388}
and Nagao~\cite{cryptoeprint:2004:161}
present algorithms for small $g$ running in time $\softO(q^{2-2/g})$,
improving on the $O(q^2)$ algorithm of~\cite{10.1007/3-540-45539-6_2},
and the single-large-prime algorithm of~\cite{10.1007/978-3-540-40061-5_5}.
This has better performance for genus 3 than square-root algorithms like
Pollard Rho,
which has expected runtime in $\softO(q^{3/2})$;
but in genus~2, Pollard Rho is more efficient, in $\softO(q)$.
Avanzi, Thériault, and Wang~\cite{avanzi2008rethinking} further
discuss security in these cases.

Smith~\cite{Smith2009} gives a method of transferring the DLP from
hyperelliptic to non-hyperelliptic genus-3 Jacobians that applies to
18.57\% of genus-3 curves;
Diem's algorithm~\cite{10.1007/11792086_38} can then be used to solve
the DLP in time $\softO(q)$. Laine and Lauter~\cite{laine2015time}
examine and improve on Diem's attack (including analysis of the
logarithmic factors, which they estimate to be $O(\log^2q)$),
but the memory requirement for their attack is high at $\softO(q^{3/4})$.
The practical results from \cite{laine2015time} suggest that
even for \(q \sim 2^{100}\),
discrete logarithms require around $2^{113}$ field multiplications
and $1.2 \times 10^{14} $TB of memory, assuming the
reduction of \cite{Smith2009} applies; if not, the algorithm
of~\cite{10.2307/40234388} would require on the order of $2^{133}$
operations.
Genus-3 hyperelliptic curves avoiding isogeny-based attacks
are constructed in~\cite{laine2015security}.

As $g$ tends to infinity,
there exist subexponential 
attacks on the DLP
using index calculus 
(for example, \cite{Enge:2002:CDL:589604.589619});
but these have no impact for fixed genus~2 and~3.

\subsection{Generating hyperelliptic Jacobians of unknown order}

Algorithm~\ref{alg:Gen-genus3-b}
(\texttt{Gen})
takes security parameters $(\lambda,\rho)$
(as in Definition~\ref{def:lambda-rho}),
and outputs a generator $P$ for a group \(G\) such that Sutherland's algorithm
running in time \(2^\lambda\) succeeds in computing \(\#G\)
with probability less than \(1/2^\rho\).
The group \(G\) is realized as (a subgroup of) a genus-3 hyperelliptic Jacobian.
Having chosen a suitable prime $p$
as a function of \((\lambda, \rho)\),
the algorithm samples a uniformly random monic irreducible degree-7
polynomial $f(x)$ in $\F_p[x]$
and polynomials \(u\) and \(v\)
such that \(\langle{u,v}\rangle\)
is the Mumford representation of a divisor class \(P\) in \(\Jac{C}(\F_p)\),
where \(C\) is the curve defined by \(y^2 = f(x)\).
Being random,
\(P\) generates a large-order subgroup of \(\Jac{C}(\F_p)\)
with high probability.

Taking $f$ to be random
makes the probability that $C$ is a ``weak'' curve
overwhelmingly small.
The order of a random Jacobian
should have a large prime factor,
protecting against Pohlig--Hellman;
the largest prime factor should not divide $q^k - 1$ for small $k$,
protecting against MOV-type attacks~\cite{Frey:1994:RCM:179607.179640};
and the group order should be prime to $p = \operatorname{char}(\F_q)$
to avoid ``anomalous curve'' attacks~\cite{Ruck:1999:DLD:312050.312068}.
Randomly-sampled Jacobians do not have special endomorphisms,
such as the efficiently-computable real multiplication
exploited for faster point counting in~\cite{AGS-ANTS},
because these special classes of
curves form positive-codimensional subspaces of the moduli space.
Recent work of Thakur~\cite{cryptoeprint:2020:348} further discusses
classes of curves to avoid. 
The security of random ideal class groups as groups of unknown order
depends on similar assumptions and heuristics~\cite{cohen1984heuristics}.

Taking $f$ irreducible over $\F_p$ ensures that
$\Jac{C}(\F_p)$ has no points of order 2.
As we will see in~\S\ref{sec:known-order}, it may be possible to
construct points of small odd order.
We could try this for a few small primes $\ell$
to eliminate $C$ with small factors in $\#\Jac{C}(\F_p)$,
but this makes no significant difference to
the probability of semismoothness of $\#\Jac{C}(\F_p)$,
and thus to Sutherland's algorithm.
Our simulations
showed that rejecting random group orders divisible by the first few
primes decreased the semismoothness probability by less than a factor of~$2$.

\begin{algorithm}[ht]
    \caption{\texttt{Gen}.  Constructs a random unknown-order (subgroup of a) genus-3 hyperelliptic Jacobian.}
    \label{alg:Gen-genus3-b}
    \KwIn{$(\lambda, \rho)$}
    \KwOut{%
        A prime $p$, a hyperelliptic genus-3 curve $C/\F_p$,
        and $P \in J_C$
        such that $\langle{P}\rangle$
        has unknown order.
    }
    Determine $n$ such that a random genus-3 curve over an $n$-bit prime field has $\lambda$-bits of security with probability $1 - 1/2^\rho$
    \;
    \(p \gets \) a random $n$-bit prime
    \;
    Sample random $u(x) = x^3 + u_2x^2 + u_1x + u_0$ in \(\F_p[x]\)
    \;
    Sample random $v(x) = v_2x^2 + v_1x + v_0$ in \(\F_p[x]\)
    \;
    \Repeat{$\gcd( f(x), f'(x)) = 1$ \textbf{and} $f$ is irreducible}{
        Sample random $w(x) = x^4 + w_3x^3 + w_2x^2 + w_1x + w_0$ in \(\F_p[x]\)
        \;
        \(f(x) \gets v(x)^2 + u(x)w(x)\)
        \;
    }
    \(P \gets \langle{u,v}\rangle\)
    \;
    \Return{\((p, C, P)\) where \(C\) is the hyperelliptic curve $y^2 = f(x)$ over \(\F_p\)}
    \;
\end{algorithm}

\hide{
Having chosen $C$, we also need an element $P$ of $\Jac{C}(\F_q)$ to serve as a generator for the group of unknown order.
This can be chosen in a similar way, choosing random coefficients for
$u(x)$ as the first polynomial in the Mumford representation $\langle u,
v \rangle$ of $P$,
before solving for $v$ such that $v^2 = f \pmod{u}$ (recovering $v$ can be done in the same way as
from the compressed form of a divisor \cite{hess2001two});
if no such $v$ exists, we try another $u$.
Of course, the order of $P$ is also unknown,
so we cannot know whether $P$ generates the full group or a subgroup.
But as we will see below, there is a high probability that a randomly
selected $P$ will not have a small order---so it will at least generate
a large-order subgroup of $\Jac{C}(\F_q)$.
}

To ensure that not even the constructor of $C$ knows $\#\Jac{C}(\F_p)$,
and that $C$ and $P = \langle{u,v}\rangle$ were indeed generated randomly,
we suggest that $u$, $v$, and $w$ be chosen by deterministic ``nothing up my
sleeve''-type process.
For example, the coefficients of $f$ might be taken
from the hash of a certain string.
Suppose this process were manipulated by taking multiple ``seeds'',
and testing each resulting curve for weakness.
If the probability of encountering a weak curve among random curves is $\delta$,
and testing for weakness costs $2^n$ operations,
then a malicious actor requires around $\delta^{-1}2^n$ operations to generate a weak $C$.
A sceptical verifier, on the other hand,
must only test the proposed $C$ just once to detect cheating,
at a cost of just $2^n$ operations.
This imbalance of the cost of cheating versus verifying is a deterrent for attackers,
regardless of the weakness in question.

\hide{
The only ``weak cases'' we are aware of are weak groups for Sutherland's algorithm or special curves for which point-counting is easy (such as CM curves, RM curves, etc).
We are choosing the polynomial $f(x)$ that defines the curve uniformly in $\F_p[x]$, where $p > 2^{300}$, so we expect the probability to hit a special curve to be negligible.
\TODO[Revisit this discussion to ensure it is rigorous enough.]
}

Now the order of the Jacobian $\Jac{C}(\F_p)$ (and the subgroup generated by $P$)
cannot feasibly be computed, not even by the party who constructed the
curve: we have achieved trustless setup.
This group can then be used in cryptographic constructions including accumulators and VDFs.
Overall, the generation of a new hyperelliptic curve is relatively cheap. Therefore, just as in the case of class groups, it should be feasible to generate a new group of unknown order for each new instance of an accumulator or VDF if desired.

Elements of \(\Jac{C}(\F_q)\) are represented as pairs of
polynomials $\langle u,v \rangle$ with $\deg{v} < \deg{u} \le g$,
so elements can be stored concretely with 6 elements of $\F_q$,
and further compressed to just 3 \(\F_q\)-elements and 3 extra bits
using the method of~\cite{hess2001two}.\footnote{
    Given that hyperelliptic Jacobians are a function-field analogue
    of ideal class groups of quadratic fields,
    it is natural to ask why the almost-ideal hyperelliptic Jacobian element compression
    algorithm of~\cite{hess2001two}
    does not have an efficient class-group analogue.
    To compress a Jacobian element \(\langle{u,v}\rangle\),
    the algorithm of~\cite{hess2001two} begins by factoring \(u\),
    a polynomial over a finite field;
    we know how to do this efficiently.
    A class-group analogue compressing \((a,b)\)
    would need to factor the integer \(a\),
    which is a much harder problem.
}
For \((\lambda,\rho) = (128,55)\),
with \(\approx 640\)-bit fields,
this means that group elements can be stored in \(\approx 1920\) bits;
elements of a class group of equivalent security
require \(\approx 2860\) bits with the compression
of~\S\ref{sec:compression} (or \(\approx 3840\) bits without it).
Moving to the more paranoid security level
of \((\lambda,\rho) = (128,128)\),
hyperelliptic Jacobian elements require \(\approx 3392\) bits
while class-group elements require \(\approx 5088\) bits
(or \(\approx 6784\) bits without the compression
of~\S\ref{sec:compression}).
We therefore claim that genus-3 Jacobians offer more compact elements
than class groups at the same security level.

\hide{
We now come to concrete suggestions for the bound $s$
in Assumption~\ref{assum:infeasible}.
While there have been no implementations
of Schoof-type algorithms in genus 3,
we can infer something about its difficulty
from work in the substantially easier, yet already hard, genus-2 case.
In view of the results discussed in~\S\ref{sec:point-counting},
taking $s = 60$ should be sufficient:
at this point, the memory and time requirements should make computation infeasible.

Let us take $s = 60$ in order to make some concrete efficiency claims.
Consider the modified PoE from Example~\ref{example:PoE-cofactor}.
The only place where $s$ affects efficiency is in the final verification step,
checking $[N]([\ell]Q + [r]U) = W$;
here, $N = 60!$ has 273 bits.
The choice of prime $\ell$ respects the security parameter, so would
perhaps be on order of $100$ or more bits; $r$ would be similar.
So the additional cost to the verifier is comparable to the existing cost.
Thus, we claim that using genus-3 Jacobians should still be faster than using ideal class groups
(even more so if similar countermeasures are applied to class groups,
as we suggest in~\S\ref{sec:low-order-class-groups}).
}

\section{
    Elements of known order
}
\label{sec:known-order}

We now briefly consider the problem of constructing points of known order
in groups of unknown order.
This will give us an idea of how to work with Jacobians
when the low-order or adaptive root assumptions are imposed.

\subsection{Low-order assumptions and cofactors}

Common additional requirements on unknown-order groups
include
\begin{itemize}
    \item
        the \emph{low-order assumption}:
        finding an element $P$ of a given order \(s\) in \(G\) is hard
        (see~\cite[Def.~1]{cryptoeprint:2018:712});
        and
    \item
        the \emph{adaptive root assumption}:
        extracting roots of a given element---that is,
        given \(Q\) and \(s\), find \(P\) such that \(Q = [s]P\) in
        \(G\)---is hard
        (see~\cite[Def.~2]{cryptoeprint:2018:712}
        and~\cite{10.1007/978-3-030-17659-4_13}).
\end{itemize}
These assumptions only make sense if the adversary must solve arbitrary
instances in a randomly-sampled $G$;
it is not possible to define security for a fixed $G$.
Example~\ref{example:PoE} describes Wesolowski's Proof of Exponentiation (PoE),
a protocol which requires these assumptions.

\begin{example}[PoE]
    \label{example:PoE}
    Let $G$ be a group, chosen according to security
    parameter $\lambda$. The Proof of Exponentiation takes as input $u$ and
    $w$ in $G$ and $x$ in $\Z$, and aims to prove that $u^x = w$
    in significantly less time than it takes to compute~$u^x$.
    It proceeds interactively as follows
    (although it can be made non-interactive:
    see~\cite{10.1007/978-3-030-17659-4_13}).
    \begin{enumerate}
        \item Verifier sends a random prime $\ell \in \text{Primes}(\lambda)$ to Prover.
        \item Prover computes $q = \lfloor x/\ell \rfloor$
            and $Q = u^q$,
            and sends $Q$ to Verifier.
        \item Verifier computes $r = x \text{ mod }\ell$, and accepts if $Q^\ell u^r = w$
    \end{enumerate}
\end{example}

To see why the security of this protocol
requires the low-order assumption,
suppose we know an element $\epsilon$ of order $2$ in $G$
(for example, if $G$ is an RSA group, then we can take \(\epsilon = -1\)).
Then for any valid proof that $u^x = w$, we can easily generate a
false proof of the contradictory statement $u^x = \epsilon w$,
by replacing $Q$ with $Q' = \epsilon Q$ in the proof.
Since $\ell$ is odd, $(Q')^\ell u^r = \epsilon Q^\ell u^r = \epsilon w$
holds despite the fact that $u^x \not= \epsilon w$.
This is why when using RSA groups, it is important to use the quotient
$(\Z/N\Z)^*/\langle{\pm1}\rangle$ to eliminate this element.


Suppose we are given an algorithm \texttt{Gen}
constructing unknown-order groups
reaching the \((\lambda,\rho)\) security level.
If we can specify a set
\(\mathcal{S}\) containing the integers \(s\)
such that we can construct elements of order \(s\)
or extract \(s\)-th roots
in groups \(G\) output by \texttt{Gen}
in \(< 2^\lambda\) operations
with probability \(> 2^{-\rho}\),
then the low-order and adaptive-root assumptions
hold in the subgroup
\begin{equation} \label{eq:subgroup}
    [S]G = \{ [S]P \mid P \in G \}
    \qquad
    \text{where}
    \qquad
    S := \operatorname{lcm}(\mathcal{S})
    \,.
\end{equation}
We will propose conservative choices for \(\mathcal{S}\)
for concrete groups below.
In the meantime, to give some concrete intuition,
if we take \(\mathcal{S} = \{1,\ldots,60\}\)
then \(S\) is an \(84\)-bit integer,
so multiplication by $S$ is efficient.

The operation of protocols such as accumulators in $[S]G$ is standard,
but some protocols may need modification:
for example,
proofs may require an extra check that an element is indeed in the group.
The issue here is that given a point $Q$ in $G$,
testing subgroup membership $Q \in [S]G$ is not easy.
However,
the original point $Q$ is effectively a proof that $[S]Q$ is in $[S]G$,
because this can be easily verified;
so $Q$ should be sent instead of $[S]Q$ in cryptographic protocols,
and the verifier can perform the multiplication by $S$ themselves.

Using $[S]G$ in place of $G$ has an impact on efficiency,
due to the extra scalar multiplications required.
This impact is highly protocol-dependent,
but in most cases only a few extra multiplications should be needed.
To give a specific example,
we revisit the PoE from Example~\ref{example:PoE} in Example~\ref{example:PoE-cofactor}.
The verifier only needs to perform one extra multiplication-by-$S$ during verification
when working in $[S]G$ instead of $G$.
We suggest that this is efficient enough for practical use,
and that other protocols using the adaptive root assumption
can be modified in a similar way.

\begin{example}[\protect{PoE in \([S]G\)}]
    \label{example:PoE-cofactor}
    We begin the PoE protocol with input
    $U \in G$, $W \in [S]G$, and $x \in \Z$.
    The claim to be proven is that $[x][S]U = W$ in $[S]G$.
    The protocol proceeds as follows:
    \begin{enumerate}
	    \item Verifier selects a random $\ell$ from
            $\text{Primes}(\lambda)\setminus\mathcal{S}$
            and sends $\ell$ to Prover.
	    \item Prover computes $q = \lfloor x/\ell \rfloor$,
            computes $Q = [q]U$ in $G$
            and sends $Q$ to Verifier.
	    \item Verifier computes $r = x \text{ mod }\ell$,
            and accepts if $Q$ is in $G$ and $[S]([\ell]Q + [r]U) = W$.
    \end{enumerate}
\end{example}

The security of this protocol depends on the choice of \(\mathcal{S}\).
Given a valid proof of $[x][S]U = W$,
in order to falsely prove $[x][S]U = W+P$,
the prover must compute $[1/\ell]P$ for the $\ell$ chosen by the verifier.
This may be possible if the order of $P$ is known,
but this is supposed to be infeasible because \(\ell\) is not in \(\mathcal{S}\).

\begin{remark}
    The impact of finding small-order elements
    is highly domain-specific.
    For example,
    in the VDF of~\cite{10.1007/978-3-030-17659-4_13,cryptoeprint:2018:712},
    even if points of known order can be found,
    forging a false PoE still requires knowing the true result of the
    exponentiation---and hence still requires computing the output of the VDF.
    Relying on a PoE would thus break the requirement that the VDF output is unique,
    but it would still provide assurance of the delay.
    For accumulators,
    we need an analogue of the strong RSA assumption rather than the adaptive root assumption:
    it should be hard to find \emph{chosen} prime roots of an element
    (recall that the membership witness of $\ell$ in $A$ is the $\ell$-th root of $A$).
    This case can be addressed differently,
    by simply disallowing the accumulation of small primes $\ell$
    dividing elements of \(\mathcal{S}\).
    Finding $\ell$-th roots with $\ell$ not in \(\mathcal{S}\)
    is supposed to be infeasible,
    so here we do not need to use $[S]G$.
\end{remark}

\subsection{Elements of known order in class groups and Jacobians}

For class groups,
it is well-known that the factorization of \(\Delta\) reveals the
2-torsion structure of \(Cl(\Delta)\),
and even allows the explicit construction of elements of order \(2\).
Similarly, for Jacobians,
if \(C/\F_q\) is defined by \(y^2 = f(x)\),
then the factorization of \(f(x)\) reveals the \(2\)-torsion structure
of \(\Jac{C}(\F_q)\), and lets us construct explicit points of order \(2\).
This motivates the restriction to negative prime \(\Delta\)
when using \(Cl(\Delta)\) as an unknown-order group,
and our restriction to irreducible \(f\)
in Algorithm~\ref{alg:Gen-genus3-b}.

Belabas, Kleinjung, Sanso, and Wesolowski~\cite{cryptoeprint:2020:1310}
give several constructions of special discriminants \(\Delta\)
together with a known ideal of small odd order in \(Cl(\Delta)\).
Similarly,
we can construct hyperelliptic Jacobians equipped with a point of small
order, as in Example~\ref{example:Jacobian-known-order}.

\begin{example}
    \label{example:Jacobian-known-order}
    Let \(\ell\) be an odd prime,
    and set \(g = (\ell-1)/2\).
    Choose a polynomial \(c(x)\) over \(\F_q\)
    of degree \(\le g\)
    such that \(f(x) := x^\ell + c(x)^2\) is squarefree.
    Then \(C: y^2 = f(x)\) is a hyperelliptic curve of genus \(g\),
    and
    the divisor $D = (0,c(0)) - (\infty)$
    represents a nontrivial element of order $\ell$ in \(\Jac{C}(\F_q)\)
    (because \(\ell D\) is the principal divisor of the function $y-c(x)$).
    Taking $\ell = 7$ gives a family of genus-3 Jacobians with a known element of order $7$.
\end{example}

These discriminants and curves generally do not occur when
\(\Delta\) or \(f\) is chosen in a ``nothing-up-my-sleeve'' way.
In any case,
the risk of choosing groups with constructible small-order elements
can be eliminated by using a smooth cofactor \(S\).

\hide{
If $E/\F_q$ is an elliptic curve,
then whether or not of the order of $E(\F_q)$ is known,
we can try to compute elements of small order $\ell$
by finding roots of the division polynomial $\Psi_\ell$ for $E$.
In the same way, we can try to construct elements of small, known order on
hyperelliptic Jacobians using their division ideals,
hence invalidating the adaptive root assumption.
}

There are three curve-specific methods for constructing elements of
known small order, or deducing information about small divisors of
the order of a given element, which do not apply to class groups.
The first is to use the division ideals.
This is practical for small primes like \(2\), \(3\), and \(5\) 
(reinforcing the need for the cofactor $S$ above).
However,
if we assume that there exists no feasible Schoof-type algorithm
for counting points on genus 3 curves,
then we implicitly assume that it is infeasible to construct
\(\ell\)-division ideals for \(\ell\) larger than some bound
that is polynomial with respect to the security parameters.

The second method is to use repeated divisions by 2 in $\Jac{C}(\F_q)$
to construct points of order $2^k$ for arbitrarily large $k$.
Since $2^k$ is coprime to all odd primes $\ell$,
this allows malicious provers to easily find $\ell$-th roots
for these points (that is, given a point $Q$, find $P$ such that $[\ell]P = Q$).
But repeated division by $2$ in $\Jac{C}(\F_q)$
requires the repeated extraction of square roots in \(\F_p\),
which quickly requires repeated quadratic field extensions and the field computations blow-up exponentially.
Using hyperelliptic curves in the form $y^2 = f(x)$ with $f(x)$ irreducible
ensures that the required square roots do not exist in $\F_p$.

Similarly, we might calculate repeated divisions by very small odd primes.
Using the group $[S]\Jac{C}(\F_q)$ will
kill off powers of these small primes dividing the group order.
It could also be possible to simply test for these repeated divisions
during the curve generation procedure, allowing parties to check for
small factors of the group order---and then kill these off with a tailored choice of $S$.
It is an interesting open problem to generate an easily
verifiable proof that a Jacobian does not have any points of low order.

The third method involves the Tate pairing.
Let $C$ be a hyperelliptic curve over~$\F_{q}$,
let $\ell$ be a prime (coprime to~$q$),
and let
$k$ be the smallest positive integer such that $\ell \mid q^k - 1$.
The reduced \(\ell\)-Tate pairing is a bilinear mapping
\begin{equation*}
    t_\ell : \Jac{C}[\ell] \times
    \Jac{C}(\F_{q^k})/\ell\Jac{C}(\F_{q^k})
    \longrightarrow
    \mu_\ell
    \subset \F_{q^k}^\times
    \,,
\end{equation*}
where $\mu_\ell$ is the group of $\ell$-th roots of unity
(see~\cite{ghv07}
for background on pairings on hyperelliptic curves).
If we can find points of known order $\ell$,
then the \(\ell\)-Tate pairing can give information about the
$\ell$-divisibility of other points.

Suppose we can find a point $Q$ in $\Jac{C}(\F_{q^k})$
of small known prime order $\ell$.
Then for any point $P$ in $\Jac{C}(\F_{q})$,
we can efficiently compute $t_{\ell}(Q,P)$ in $\mu_{\ell}$
(assuming $k$ is only polynomially large in $\log q$).
Now,
if $\ell \nmid \operatorname{Ord}(P)$, then $P = \ell P'$ for some $P'$, so
$t_\ell(Q,P) = 1$.
By the contrapositive,
if $t_\ell(Q,P) \not= 1$,
then $\ell$ divides the order of $P$.

Unfortunately, the converse is not so simple:
$t_\ell (Q,P) = 1$ for a single point $Q$ of order $\ell$
does not imply \(\ell \nmid \operatorname{Ord}(P)\).
Instead, it must be shown that
$t_\ell (Q,P) = 1$ for \emph{all} $Q$ in $\Jac{C}[\ell]$.
Thus, we require a basis $\{Q_1,\ldots,Q_{2g}\}$ of $\Jac{C}[\ell]$
which we can test:
if $t_{\ell}(Q_i, P) = 1$ for $1 \le i\le 2g$,
then the bilinearity of the Tate pairing implies
$t_{\ell}(Q,P)$ for all $Q$ in $\Jac{C}[\ell]$,
and hence that $\gcd(\operatorname{Ord}(P), \ell) = 1$.

The utility of this approach
is limited by the difficulty of constructing points of order \(\ell\),
but also by the field extension degree \(k\)
(since the coordinates of \(Q\) and the value of \(t_\ell(Q,P)\)
are in \(\F_{q^k}\));
and \(k\), being the order of \(q\) modulo \(\ell\),
tends to blow up with \(\ell\).
If \(q\) is well-chosen,
then in practice we can learn very little information about the orders
of random points in \(\Jac{C}(\F_q)\),
or any information at all for points in $[S]\Jac{C}(\F_q)$ for a
suitable $S$.
For the Jacobian case, we conjecture that $\mathcal{S} = \{1, \ldots, 60\}$
is sufficient for a $(128,128)$ security level, based on the discussion in \S\ref{sec:point-counting}.
For class groups, $\mathcal{S}$ can either be empty in the case of a prime
discriminant, or $\mathcal{S} = \{ 2 \}$ if a non-prime discriminant is used.

\printbibliography 

\appendix

\end{document}